\documentclass[runningheads]{llncs}
\usepackage{amsmath}
\usepackage{amssymb}
\usepackage{graphicx}
\usepackage{graphics}
\usepackage{listings}

\usepackage[algo2e,ruled,vlined]{algorithm2e}

\title{A reduction of the logspace shortest path problem to biconnected graphs}
\titlerunning{Logspace shortest path algorithms}

\author{Boris Brimkov}
\authorrunning{B. Brimkov}

\institute{
Computational \& Applied Mathematics, Rice University, Houston, TX 77005 USA\\
\email{boris.brimkov@rice.edu}}

\begin{document}
\maketitle

\begin{abstract}
In this paper, we reduce the logspace shortest path problem to biconnected graphs; in particular, we present a logspace shortest path algorithm for general graphs which uses a logspace shortest path oracle for biconnected graphs. We also present a linear time logspace shortest path algorithm for graphs with bounded vertex degree and biconnected component size, which does not rely on an oracle. The asymptotic time-space product of this algorithm is the best possible among all shortest path algorithms.
\medskip
 

\end{abstract}
\section{Introduction}

The logspace computational model entails algorithms which use a read-only input array and $O(\log n)$ working memory. 
For general graphs, there is no known deterministic logspace algorithm for the shortest path problem. In fact, the shortest path problem is NL-complete, so the existence of a logspace algorithm would imply that L=NL \cite{jakoby1}. In this paper, we reduce the logspace shortest path problem to biconnected graphs, and present a linear time logspace shortest path algorithm for parameter-constrained graphs.

An important result under the logspace computational model which is used in the sequel is Reingold's deterministic polynomial time algorithm \cite{reingold} for the undirected $st$-connectivity problem (USTCON) of determining whether two vertices in an undirected graph belong to the same connected component. There are a number of randomized logspace algorithms for USTCON (see, for example, \cite{barnes_feige,feige,kosowski}) which perform faster than Reingold's algorithm but whose output may be incorrect with a certain probability. 

There are also a number of logspace algorithms for the shortest path problem and other graph problems on special types of graphs (see \cite{asano6,brimkov2,ktrees,jakoby1,munro_ramirez}). 
As a rule, due to time-space trade-off, improved space-efficiency is achieved on the account of higher time-complexity. Often the trade-off is rather large, yielding time complexities of $O(n^c)$ ``for some constant $c$ significantly larger than 1" \cite{jakoby2}. In particular, the time complexity of Reingold's USTCON algorithm remains largely uncharted but is possibly of very high order. The linear time logspace shortest path algorithm presented in this paper avoids this shortcoming, at the expense of some loss of generality. In fact, its time (and space) complexity is the best possible, since a hypothetical sublinear-time algorithm would fail to print a shortest path of length $\Omega(n)$.

This paper is organized as follows. In the next section, we recall some basic definitions and introduce a few concepts which will be used in the sequel. In Section~3, we present a reduction of the logspace shortest path algorithm to biconnected graphs. In Section~4, we present a linear time logspace algorithm for parameter-constrained graphs. We conclude with some final remarks in Section~5.

\section{Preliminaries}

A \emph{logspace algorithm} is an algorithm which uses $O(\log n)$ working memory, where $n$ is the size of the input. In addition, the input and output are respectively read-only and write-only, and do not count toward the space used. The \emph{shortest path problem} requires finding a path between two given vertices $s$ and $t$ in a graph $G$, such that the sum of the weights of the edges constituting the path is as small as possible. 
In general, if $s$ and $t$ are not in the same connected component, or if the connected component containing $s$ and $t$ also contains a negative-weight cycle, the shortest path does not exist. For simplicity, we will assume there are no negative-weight cycles in $G$, although the proposed algorithms can be easily modified to detect (and terminate at) such cycles without any increase in overall complexity. We will also assume that $G$ is encoded by its adjacency list, where vertices are labeled with the first $n$ natural numbers. The $j^{\text{th}}$ neighbor of vertex $i$ is accessed with $Adj(i,j)$ in $O(1)$ time, and \emph{degree}$(i)=|Adj(i)|$.

An \emph{articulation point} in $G$ is a vertex whose deletion increases the number of connected components of $G$. A \emph{block} is a maximal subgraph of $G$ which has no articulation points; if $G$ has a single block, then $G$ is \emph{biconnected}. The \emph{block tree} $T$ of $G$ is the bipartite graph with parts $A$ and $B$, where $A$ is the set of articulation points of $G$ and $B$ is the set of blocks of $G$; $a\in A$ is adjacent to $b \in B$ if and only if $b$ contains $a$. We define the \emph{id} of a block in $G$ to be 
$(\emph{largest, smallest})$,
where \emph{largest} and \emph{smallest} are the largest and smallest vertices in the block with respect to their labeling from 1 to $n$. Clearly, each block in $G$ has a unique \emph{id}. Note also that it is possible to lexicographically compare the $id$s of two or more blocks, i.e., if $id_1=(\ell_1,s_1)$ and $id_2=(\ell_2,s_2)$, then $id_1>id_2$ if $\ell_1> \ell_2$ or if $\ell_1= \ell_2$ and $s_1>s_2$.

Given numbers $a_1$, $a_2$, and $p$, we define the \emph{next} number after $p$ as follows:

\begin{equation*}
next(a_1,a_2,p)=\begin{cases}
a_1&\text{ if }a_2\leq p<a_1\text{ or }p<a_1\leq a_2\text{  or }a_1 \leq a_2 \leq p\\
a_2&\text{  otherwise.}
\end{cases}
\end{equation*}


\noindent We extend this definition to a list $L$ of not necessarily distinct numbers by defining the \emph{next} number in $L$ after $p$ to be a number in $L$ larger than $p$ by the smallest amount, or if no such number exists, to be the smallest number in $L$. The \emph{next} number in $L$ can be found with logspace and $O(|L|)$ time, given sequential access to the elements of $L$, by repeatedly applying the \emph{next} function. 

%


\section{Reducing the logspace shortest path problem to biconnected graphs}


Let \emph{connected}$(H;v_1,v_2)$ be an implementation of Reingold's USTCON algorithm which takes in two vertices of a graph $H$ and returns \emph{true} if they belong to the same connected component, and \emph{false} otherwise. Let \emph{pathInBlock}$(H;v_1,v_2)$ be a polynomial time, logspace oracle which takes in two vertices of a biconnected graph $H$ and prints the shortest path between them.

Clearly, the encoding of a graph $H$ can be reduced with logspace and polynomial time to the encoding of some induced subgraph $H[S]$. Thus, by transitivity and closure of reductions, the functions \emph{connected}$(H[S];v_1,v_2)$ and \emph{pathInBlock}$(H[B];v_1,v_2)$ can be used with logspace and polynomial time, where $S$ and $B$ are sets of vertices computed at runtime and $H[B]$ is biconnected.

The \emph{connected} function reduces the logspace shortest path problem to connected graphs. In this section, we will further reduce this problem to biconnected graphs, by presenting a logspace algorithm for finding the shortest path between two vertices in an arbitrary graph using the oracle \emph{pathInBlock}. 


\subsection{Constructing a logspace traversal function}

Let $G$ be a graph of order $n$, and $v_1$ and $v_2$ be two vertices that belong to the same block; the set of all vertices in this block will be referred to as \emph{block}$(v_1,v_2)$. Using the \emph{connected} function, is easy to construct a logspace function \emph{isInBlock}$(v_1,v_2,v)$ which returns \emph{true} when $v$ is part of \emph{block}$(v_1,v_2)$ and \emph{false} otherwise; see Table 1 for pseudocode. This procedure can be used to access the vertices in $block(v_1,v_2)$ sequentially. A similar procedure \emph{areInBlock}$(u,v)$ can be defined which returns \emph{true} when $u$ and $v$ are in the same block, and \emph{false} otherwise.

A vertex of $G$ is an articulation point if and only if two of its neighbors are not in the same block. Thus, using the \emph{isInBlock} function, we can construct a function \emph{isArticulation}$(v)$ which returns \emph{true} when $v$ is an articulation point and \emph{false} otherwise; see Table 1 for pseudocode. We also define the function $id(v_1,v_2)$, which goes through the vertices of \emph{block}$(v_1,v_2)$ and returns (\emph{largest, smallest}), where \emph{largest} and \emph{smallest} are respectively the largest and smallest vertices in \emph{block}$(v_1,v_2)$ according to their labeling.

Let $p$ be an articulation point\footnote{The subsequent definitions and functions remain valid when $p$ is not an articulation point, and can be used in special cases, e.g., when $G$ only has one block.} in \emph{block}$(v_1,v_2)$. To find the \emph{next} articulation point in \emph{block}$(v_1,v_2)$ after $p$, we can create a function \emph{nextArticulation}$(v_1,v_2,p)$ which uses each articulation point in \emph{block}$(v_1,v_2)$ as a member of list $L$ and applies the \emph{next} function. Note that the vertices in $L$ do not have to be stored, but can be generated one at a time; see Table 1 for pseudocode. Similarly, to identify the block containing $p$ and having the \emph{next} $id$ after $id(v_1,v_2)$, we can create a function \emph{nextBlock}$(v_1,v_2,p)$ which uses the \emph{id}s of the blocks identified by $p$ and each of its neighbors as members of a list $L$ and applies the \emph{next} function. Note that the \emph{id}s in $L$ do not have to be stored but can be computed one at a time; see Table 1 for pseudocode.

Finally, given articulation point $p$ and vertex $v$ in the same block, we will call the component of $G-\{block(v,p)\backslash\{p\}\}$ which contains $p$ the \emph{subgraph of G rooted at block$(v,p)$ containing p}, or \emph{subgraph}$(v,p)$. This subgraph can be traversed with logspace by starting from $p$ and repeatedly moving to the \emph{next} block and to the \emph{next} articulation point until the starting block is reached again. This procedure indeed gives a traversal, since it corresponds to visiting the \emph{next} neighbor in the block tree $T$ of $G$, which generates an Euler subtour traversal (cf. \cite{tarjan_vishkin}). In addition, during the traversal of \emph{subgraph}$(v,p)$, each vertex can be compared to a given vertex $t$, in order to determine whether the subgraph contains $t$. Thus, we can create a function \emph{isInSubgraph}$(v,p,t)$ which returns \emph{true} if $t$ is in \emph{subgraph}$(v,p)$ and \emph{false} otherwise; see Table 1 for pseudocode.

\begin{algorithm2e}
\NoCaptionOfAlgo
\caption{\textbf{Table 1:} List of subroutines}

\SetKwProg{Fn}{function}{}{end}
\Fn{isInBlock$(v_1,v_2,v)$}{
\textbf{if} $v=v_1$ \textbf{or} $v=v_2$ \textbf{then return} \emph{true}\;

\For{$i=1$ \emph{\KwTo} $n$, $i\neq v_1$, $i\neq v$}{
\textbf{if} $\lnot \emph{connected}(G-i;v_1,v)$ \textbf{then return} \emph{false}\;
}

\textbf{if} $\lnot \emph{connected}(G-v_1;v_2,v)$ \textbf{then return} \emph{false}\;
\Return{true}\;
}

\SetKwProg{Fn}{function}{}{end}
\Fn{isArticulation$(v)$}{
\For{$i=2$ \emph{\KwTo} degree$(v)$}{
\textbf{if} $\lnot \emph{isInBlock}(v,Adj(v,1),Adj(v,i))$ \textbf{then return} \emph{true}\;
}
\Return{false}\;
}

%

%
%
%
%
%

\SetKwProg{Fn}{function}{}{end}
\Fn{nextArticulation$(v_1,v_2,p)$}{
  $a=p$\;

\For{$v=1$ \emph{\KwTo} $n$}{
\textbf{if} \emph{isInBlock}$(v_1,v_2,v)$ \textbf{and} \emph{isArticulation}$(v)$ \textbf{then} $a=next(a,v,p)$;
}
\Return{$a$}\;

}

\SetKwProg{Fn}{function}{}{end}
\Fn{nextBlock$(v_1,v_2,p)$}{
  $id^*=id(v_1,v_2)$; $a=id(p,Adj(p,1))$; $i^*=1$\;
  \For{$i=2$ \emph{\KwTo} degree$(p)$}{
    $b=id(p,Adj(p,i))$\;
\textbf{if} $a\leq id^*<b$ \textbf{or} $id^*<b\leq a$ \textbf{or}  $b \leq a \leq id^*$ \textbf{then} $a=b$; $i^*=i$\;

  }
\Return{$Adj(p,i^*)$}\;
}


\SetKwProg{Fn}{function}{}{end}
\Fn{isInSubgraph$(v,p,t)$}{

$v_1=p$; $v_2=v$; $id^*=id(v_1,v_2)$; $v_2=\emph{nextBlock}(v_1,v_2,p)$\;

\While{$id(v_1,v_2)\neq id^*$}{
$p=\emph{nextArticulation}(v_1,v_2,p)$; $v=\emph{nextBlock}(v_1,v_2,p)$\;
\textbf{if} \emph{isInBlock}$(v_1,v_2,t)$ \textbf{then return} \emph{true}\;
$v_1=p$; $v_2=v$\;
} 
\Return{false}\;
}

%
%

\end{algorithm2e}

\subsection{Main Algorithm}


Using the subroutines outlined in the previous section and the oracle \emph{pathInBlock}, we propose the following logspace algorithm for finding the shortest path in a graph $G$. The main idea is to print the shortest path one block at a time by locating $t$ in one of the subgraphs rooted at the current block.

\begin{algorithm2e}[H]
\SetKwFor{Repeat}{repeat}{}{end}
\SetKwFor{When}{when}{do}{}

\SetKwProg{Fn}{function}{}{end}
\Fn{\emph{shortestPath}$(G;s,t)$}{
\textbf{if} $\neg connected(G;s,t)$ \textbf{then return} ``Path does not exist"\;
$current=s$\;

\textbf{\emph{start}}\\
\If{$areInBlock(current,t)$}{ \Return{ $pathInBlock(G[block(current,t)];current,t)$\;}}

\For{$i=1$ \emph{\KwTo} $degree(current)$} {

\For{$potential = 1$ \emph{\KwTo} $n$}{ 
  \If{ isInBlock$($current, $Adj(\text{current},i)$, potential$)$ \emph{\textbf{and}}  $isArticulation(potential)$ \emph{\textbf{and}} $isInSubgraph(current,potential,t)$}{ 
$pathInBlock(G[block(current,potential)];current,potential)$;
$current=potential$\;
\textbf{goto \emph{start}}\;
}
}
}
}
\NoCaptionOfAlgo
\caption{\textbf{Algorithm 1:} Shortest path using oracle}
\end{algorithm2e}

\begin{theorem}
Algorithm 1 finds the correct shortest path between vertices $s$ and $t$ in graph $G$ with logspace and polynomial time, using a shortest path oracle for biconnected graphs.
\end{theorem}
\begin{proof}

Let $p_0=s$ and $p_{\ell+1}=t$; the shortest path between $p_0$ and $p_{\ell+1}$ is $P=p_0 P_0 p_1 P_1\ldots p_{\ell}P_{\ell} p_{\ell+1}$, where $p_1,\ldots,p_{\ell}$ are articulation points and $P_0,\ldots,P_{\ell}$ are (possibly empty) subpaths which contain no articulation points. Let $b_i=block(p_i,p_{i+1})$ for $0\leq i \leq \ell$, so that \emph{pathInBlock}$(G[b_i];p_i,p_{i+1})=p_iP_ip_{i+1}$.

Suppose the subpath $p_0 P_0 \ldots p_i$, $i\geq 0$, has already been printed and that the vertex $p_i$ is stored in memory. In each iteration of the main loop, the function \emph{isInSubgraph}$(p_i,p,t)$ returns \emph{true} only for $p=p_{i+1}$ when run for all articulation points $p$ in all blocks containing $p_i$. The function \emph{pathInBlock}$(G[b_i],p_i,p_{i+1})$ is then used to print $P_{i+1}$ and $p_{i+1}$. Finally, $p_i$ is replaced in memory by $p_{i+1}$, and this procedure is repeated until $p_{\ell+1}$ is reached. Since the main loop is entered only if the shortest path is of finite length, the algorithm terminates, and since each subpath printed is between two consecutive articulation points of $P$, the output of Algorithm 1 is the correct shortest path between $s$ and $t$.

Since the \emph{connected} function is logspace, the \emph{isInBlock}, \emph{isArticulation} and \emph{isInSubgraph} functions are each logspace. Only a constant number of variables, each of size $O(\log n)$, are simultaneously stored in Algorithm 1, and every function call is to a logspace function (assuming the \emph{pathInBlock} oracle is logspace); thus, the space complexity of Algorithm 1 is $O(\log n)$. Note that since the vertices in \emph{block}$(v_1,v_2)$ cannot be stored in memory simultaneously, a call to the function \emph{pathInBlock}$(G[block(v_1,v_2)],v_1,v_2)$ needs to be realized by a logspace reduction, i.e., the vertices $v_1$ and $v_2$ are stored, and whenever the function \emph{pathInBlock} needs to access an entry of the adjacency list of $G[V(block(v_1,v_2))]$, it recomputes it by going through the vertices of $G$ and using the function \emph{isInBlock}.

Similarly, since the \emph{connected} function uses polynomial time, the \emph{isInBlock}, \emph{isArticulation} and \emph{isInSubgraph} functions each use polynomial time. The main loop is executed at most $O(n)$ times, and each iteration calls a constant number of polynomial time functions (assuming the \emph{pathInBlock} oracle uses polynomial time); thus, the time complexity of Algorithm 1 is $O(n^c)$ for some constant $c$. 
$\square$
\end{proof}

\section{Linear time logspace algorithm for parametrically constrained graphs}

Let \emph{BellmanFord}$(H;v_1,v_2)$ be an implementation of the Bellman-Ford shortest path algorithm \cite{bellman_ford} which takes in two vertices of a graph $H$ and prints out the shortest path between them. Let \emph{HopcroftTarjan}$(H)$ be an implementation of Hopcroft and Tarjan's algorithm \cite{hopcroft_tarjan} which returns all blocks and articulation points of a graph $H$. If the size of $H$ is bounded by a constant, \emph{BellmanFord} and \emph{HopcroftTarjan} can each be used with constant time and  a constant number of memory cells.

Let $G$ be a graph of order $n$ with maximum vertex degree $\Delta$ and maximum biconnected component size $k$. We will regard $\Delta$ and $k$ as fixed constants, independent of $n$. Using these constraints and some additional computational techniques, we will reformulate Algorithm 1 as a linear-time logspace shortest path algorithm which does not rely on an oracle. Asymptotically, both the time and space requirements of this algorithm are the best possible and cannot be improved; see Corollary 1 for more information.

\subsection{Constructing a linear time logspace traversal function}

By the assumption on the structure of $G$, the number of vertices at distance at most $k$ from a specified vertex $v$ is bounded by $\lfloor \frac{\Delta^{k+1}-1}{\Delta-1}\rfloor$. Thus, any operations on a subgraph induced by such a set of vertices can be performed with constant time and a constant number of memory cells, each with size $O(\log n)$; note that since each vertex of $G$ has a bounded number of neighbors, $G[S]$ can be found in constant time for any set $S$ of bounded size. In particular, we can construct a function \emph{blocksContaining}$(v)$ which uses \emph{HopcroftTarjan} to return all blocks containing a given vertex $v$ and all articulation points in these blocks; see below for pseudocode.


\begin{algorithm2e}
\SetKwProg{Fn}{function}{}{end}
\Fn{blocksContaining$(v)$}{

$S=\{v\}$\;
\textbf{for} $i=1$ \textbf{to} $k$ \textbf{do} $S=\bigcup_{v\in S} Adj(v)$\;

$(\mathcal{B},\mathcal{A})=$ \emph{HopcroftTarjan}$(G[S])$\;
\tcp{$\mathcal{B}$ is set of blocks, $\mathcal{A}$ is set of articulation points of $G[S]$}
\Return $\left(B=\{b\in \mathcal{B}:b\cap Adj(v)\neq \emptyset\},A=\mathcal{A}\cap \bigcup_{b\in B}b\right)$\;
}
\NoCaptionOfAlgo
\caption{\textbf{Subroutine:} Finding all blocks containing $v$, and their articulation points}
\end{algorithm2e}

Using the set of blocks and articulation points given by the \emph{blocksContaining} function, we can define functions \emph{isInBlock}$(v_1,v_2,v)$, \emph{areInBlock}$(u,v)$, \emph{isArticulation}$(v)$, \emph{id}$(v_1,v_2)$, \emph{nextArticulation}$(v_1,v_2,p)$, and \emph{nextBlock}$(v_1,v_2,p)$ analogous to the ones described in Section 3, each of which uses $O(\log n)$ space and $O(1)$ time. We can also construct an analogue of \emph{isInSubgraph}$(v,p,t)$, which uses time proportional to the size of \emph{subgraph}$(v,p)$; in particular, the time for traversing the entire graph $G$ via an Euler tour of its block tree is $O(n)$ (provided $G$ is connected) since there are $O(n)$ calls to the \emph{nextArticulation} function and $O(n)$ calls to the \emph{nextBlock} function. 

Finally, it will be convenient to define the following functions: \emph{adjacentPoints}$(v_1,v_2)$ which returns the set of articulation points belonging to blocks containing $v_2$ but not $v_1$ if $v_1\neq v_2$ and the set of articulation points belonging to blocks containing $v_2$ if $v_1=v_2$ (this function is  slight modification of \emph{blocksContaining}); \emph{last}$(L)$ which returns the last element of a list $L$; \emph{traverseComponent}$(s,t)$ which traverses the component containing a vertex $s$ and returns \emph{true} if $t$ is in the same component and \emph{false} otherwise (this function is identical to \emph{isInSubgraph}, with a slight modification in the stopping condition).


\subsection{Linear time logspace shortest path algorithm}

We now present a modified version of Algorithm 1, which uses the subroutines outlined in the previous section as well as some additional computational techniques such as ``simulated parallelization" (introduced by Asano et al. \cite{asano6}) aimed at reducing its runtime.

\begin{algorithm2e}[H]
\SetKwFor{Repeat}{repeat}{}{end}
\SetKwFor{When}{when}{do}{}

\SetKwProg{Fn}{function}{}{end}
\Fn{\emph{shortestPath}$(G;s,t)$}{
\textbf{if} $\neg$ \emph{traverseComponent}$(s,t)$ \textbf{then return} ``Path does not exist"\;
$\emph{previous}=s$; $\emph{current}=s$\;

\emph{\textbf{start}}\\
\If{areInBlock$(\text{current},t)$}{ 
\Return{BellmanFord$(G[block(current,t)];current,t)$\;}}

\For{$potential \in L:=\text{adjacentPoints}(previous,current)$ \textbf{\emph{using two pointers in serial}}}{ 
  \If{$potential=\text{last}(L)$ \emph{\textbf{or}} isInSubgraph$(current,potential,t)$}{
\emph{BellmanFord}$(G[block(current,potential)];current,potential)$;
$previous=current$; $current=potential$\;
\textbf{goto \emph{start}}\;
}
}
}
\NoCaptionOfAlgo
\caption{\textbf{Algorithm 2:} Shortest path in parameter constrained graph}
\end{algorithm2e}

\begin{theorem}
Algorithm 2 finds the correct shortest path between vertices $s$ and $t$ in graph $G$ with bounded degree and biconnected component size using logspace and linear time. 
\end{theorem}
\begin{proof}

Using the notation in the proof of Theorem 1, suppose the subpath $p_0P_0\ldots p_i$ has already been printed; $p_{\ell+1}$ cannot be in \emph{subgraph}$(p_i,p_{i-1})$, so there is no need to run \emph{isInSubgraph}$(p_i,p_{i-1},t)$. Thus, \emph{adjacentPoints}$(p_{i-1},p_i)$ is the set of feasible articulation points. Moreover, if $p_{\ell+1}$ is not in \emph{subgraph}$(p_i,p)$ for all-but-one feasible articulation points, then the last of these must be $p_{i+1}$ and there is no need to run \emph{isInSubgraph}$(p_i,p_{i+1},t)$. Finally, two subgraphs rooted at $block(p_{i-1},p_i)$ can be traversed concurrently with the technique of simulated parallelization: instead of traversing the feasible subgraphs one-after-another, we maintain two copies of the \emph{isInSubgraph} function and use them to simultaneously traverse two subgraphs. We do this in serial (without the use of a parallel processor) by iteratively advancing each copy of the function in turn; if one subgraph is traversed, the corresponding copy of the function terminates and another copy is initiated to traverse the next unexplored subgraph. Thus, Algorithm 2 is structurally identical to Algorithm 1\footnote{Indeed each of the described modifications can be implemented in Algorithm~1 as well, but would not make a significant difference in its time complexity.} and prints the correct shortest path between $s$ and $t$.

Only a constant number of variables, each of size $O(\log n)$, are simultaneously used in Algorithm 2, and every function call is to a logspace function; moreover, keeping track of the internal states of two logspace functions can be done with logspace, so the space complexity of Algorithm 2 is $O(\log n)$. 


Finally, to verify the time complexity, note that by traversing two subgraphs at once, we can deduce which subgraph contains $t$ in the time it takes to traverse all subgraphs which do \emph{not} contain $t$ or $s$. Thus, each subgraph rooted at \emph{block}$(p_i,p_{i+1})$, $0\leq i\leq \ell$, which does not contain $t$ or $s$ will be traversed at most once, so the time needed to print the shortest path is of the same order as the time needed to traverse $G$ once. 
$\square$
\end{proof}

\begin{corollary}
The time and space complexity of Algorithm 2 is the best possible for the class of graphs considered.
\end{corollary}

\begin{proof}
Let $G$ be a graph of order $n$; the shortest path between two vertices in $G$ may be of length $\Omega(n)$ so any shortest path algorithm will require at least $\Omega(n)$ time to print the path. Moreover, a pointer to an entry in the adjacency list of $G$ has size $\Omega(\log n)$, so printing each edge of the shortest path requires at least $\Omega(n)$ space. $\square$
\end{proof}


\section{Conclusion}

We have reduced the logspace shortest path problem to biconnected graphs using techniques such as computing instead of storing, transitivity of logspace reductions, and Reingold's USTCON result. We have also proposed a linear time logspace shortest path algorithm for graphs with bounded degree and biconnected component size, using techniques such as simulated parallelization and constant-time and -space calls to functions over graphs with bounded size. 

\sloppypar Future work will be aimed at further reducing the logspace shortest path problem to triconnected graphs using SPQR-tree decomposition, and to $k$-connected graphs using branch decomposition or the decomposition of Holberg \cite{decomposition}. Another direction for future work will be to generalize Algorithm 2 by removing or relaxing the restrictions on vertex degree and biconnected component size.

\section*{Acknowledgements}

This material is based upon work supported by the National Science Foundation under Grant No. 1450681.

\end{document}